\documentclass[letterpaper, 10 pt, conference]{ieeeconf}  

\IEEEoverridecommandlockouts                              
\usepackage{graphicx}   
\usepackage{algorithmic}
\usepackage{algorithm}
\usepackage{array}
\usepackage{textcomp}
\usepackage{stfloats}
\usepackage{url}
\usepackage{verbatim}
\usepackage{amssymb}
\usepackage{glossaries}

\newtheorem{theorem}{Theorem}

\newtheorem{lemma}{Lemma}
\newtheorem{remark}{Remark}
\newtheorem{assumption}{Assumption}

\newacronym{DRL}{DRL}{Deep Reinforcement Learning}
\newacronym{HJB}{HJB}{Hamilton-Jacobi-Bellman}
\newacronym{MIMO}{MIMO}{multi-input multi-output}
\newacronym{SAC}{SAC}{soft actor critic}
\newacronym{ADP}{ADP}{adaptive dynamic programming}
\newacronym{NN}{NN}{neural network}
\newacronym{DOBC}{DOBC}{disturbance observer-based control}
\newacronym{LS}{LS}{least square}
\newacronym{NDOBC}{NDOBC}{nonlinear disturbance observer-based control}
\newacronym{DOBOTC}{DOBOTC}{disturbance observer-based optimal tracking control}
\newacronym{ESO}{ESO}{extended state observer}
\newacronym{LQR}{LQR}{linear quadratic regulator}
\newacronym{RL}{RL}{reinforcement learning}
\newacronym{RLC}{RLC}{reinforcement learning control}
\newacronym{DO}{DO}{disturbance observer}
\newacronym{NDO}{NDO}{nonlinear disturbance observer}

\title{\LARGE \bf
Disturbance observer-based tracking control for roll-to-roll slot die coating systems under gap and pump rate disturbances
}

\author{Zezhi Tang$^{\ast}$, Christopher Passmore, Andrew I Campbell, Jonathan Howse, J Anthony Rossiter,\\Stephen Ebbens,  George Panoutsos
\thanks{Z. Tang is with Department of Computer Science, University College London, London, WC1E 6EA, United Kingdom. (email: zezhi.tang@ucl.ac.uk)}
\thanks{J.A. Rossiter and G. Panoutsos are with Department of Automatic Control and Systems Engineering, University of Sheffield, Sheffield, S1 3JD, United Kingdom. (emails: j.a.rossiter@sheffield.ac.uk, g.panoutsos@sheffield.ac.uk)}%
\thanks{C. Passmore, A. Campbell, J. Howse and S. Ebbens are with the Department of Chemical and Biological Engineering, University of Sheffield, Sheffield, S1 4LZ, United Kingdom. (emails: cgpassmore1@sheffield.ac.uk, a.i.campbell@sheffield.ac.uk, j.r.howse@sheffield.ac.uk,  s.ebbens@sheffield.ac.uk}%
\thanks{This work was supported by the UK Engineering and Physical Sciences Research Council. (EPSRC, grant no. EP/V051261/1)}
\thanks{Z. Tang and C. Passmore contributed equally to this work.}
\thanks{*Corresponding author.}
}

\begin{document}

\maketitle
\thispagestyle{empty}
\pagestyle{empty}

\begin{abstract}

Roll-to-roll slot die coating is a widely used industrial manufacturing technique applied in a diverse range of applications such as the production of lithium-ion batteries, solar cells and optical films. The efficiency of roll-to-roll slot die coating depends on the precise control of various input parameters such as pump rate, substrate velocity and coating gap. However, these inputs are sensitive to disturbances in process conditions, leading to inconsistencies in the various characteristics of the produced film. To address this challenge, a \gls{DO} is utilized for detecting disturbances, which may occur in the same or different channels as the control signal within the system. A generalized compensator is then implemented to mitigate the impact of these disturbances on the output, thereby enhancing uncertainty suppression. Additionally, integrating the disturbance rejection system with an output tracking controller enables the coating system to maintain the desired thickness under varying input conditions and disturbances. The effectiveness of this approach is then validated using a test rig equipped with a camera system, which facilitates the development of a data-driven model of the dynamic process, represented by state-space equations. The simulation results were demonstrated to showcase the effectiveness of the \gls{DOBOTC} system, which provides a resilient solution for the output tracking issue in a data-driven model with generalized disturbances.

\end{abstract}

\section{Introduction}

Roll-to-roll slot die coating is a crucial manufacturing method for applications such as lithium-ion batteries, solar cells and optical films \cite{ding2016review}. It is notable for its potential high coating uniformity across large areas, high throughput due to it's ability to coat at speeds greater than 600 m/s and high material utilisation \cite{romero2008response,Li2020Review,Whitaker2018Scalable}. This method, which involves extruding the coating liquid through the space between the slot die lip and a swiftly moving substrate, forming a coating bead, has been vital in ensuring the production of uniform films under steady-state, two-dimensional flow conditions \cite{lee2015response}. The operational parameters that enable such flows, defined within the concept of the coating window, are vital for ensuring the desired steady-state operation and determining the suitability of the coating method for specific production rates and applications \cite{carvalho2000low}.

The coating bead, a liquid meniscus bounded by the slot die lips, the substrate, and the menisci, plays a vital role in the uniform production of films. The steadiness of the coating flow, which needs to be two-dimensional and stable, is paramount for successful uniform film production \cite{lee2015analysis}. Also, the coating window, which refers to the range of operating conditions that allow the coating bead to generate such flows, has been a central point of numerous studies, focusing on computational and experimental methods to identify such windows in steady-state coating flows \cite{romero2008response}. 

However, despite the advantages that slot die coating offers, its practical implementation faces significant challenges. In a real-world manufacturing environment, the process is often subjected to small-scale disturbances, such as periodic variations in the coating gap, vacuum pressure, web speed, and flow rate, typically emanating from mechanical elements like gears, shafts, and rolls in the processing units \cite{romero2004low}. Therefore, in order to attenuate the influence of disturbances that exist in the system such as the slot coating process, diverse control algorithms have been discussed over the past decades. 
Moreover, tracking control has become a crucial technique to enhance precision and uniformity in film production. This development is particularly important given the challenges related to high-speed liquid coating processes. The tracking control dynamics in tensioned-web-over-slot die (TWOSD) coating are investigated in \cite{nam2010flow}, focusing on how elastohydrodynamic interactions help in tracking and maintaining a minimal gap between the substrate and coating die lip. An image analysis algorithm for tracking contact lines in slot coating flows is developed in \cite{hong2017automatic}, enabling more accurate control and adjustment of the coating process in response to various operational disturbances. However, in general terms, the research on feedback tracking control in the slot coating system, particularly within dynamic stabilization processes in production, remains limited.

Stepping beyond slot die coating applications and looking at control loop design in general, some control approaches utilize feedback control to eliminate disturbances rather than employing feedforward compensation. In such cases, the methodology behind those approaches is to reject the uncertainties via the error between the measured output and desired setpoints \cite{yang2018periodic, chen2015disturbance, tang2024output}. 
Meanwhile,  in order to further mitigate the impact of disturbances, a feedforward compensation scheme based on estimations of the uncertainties can also be utilized to counteract their effect in advance, thereby enhancing the overall robustness of the system \cite{li2014disturbance}. Furthermore, feedforward control methods such as \gls{DOBC} may be superior in offering faster responses and being less conservative \cite{tang2019control}.  For that reason, DOBC has attracted much attention in the control community \cite{bai2025deep}. For example, in \cite{tang2016unmatched}, the authors address the disturbance rejection problem for active magnetic bearing systems via the DOBC approach, and a discussion of adaptive affine formation manoeuvre control problems in the presence of external disturbances is presented in \cite{luo2020adaptive}.

To address the aforementioned challenges within slot die coating, this paper seeks to utilise the feedback and feedforward design developments from other areas.   Specifically, this paper develops a \gls{DO}-based optimal tracking control scheme for the slot coating process. The proposed method utilizes the optimal control framework to optimize the energy input to the data-driven model. The dynamic model is derived using a system identification method, with the output data obtained from images with image processing techniques. The images are captured by the camera system in the experimental platform. In addition, by introducing a disturbance attenuation scheme that serves to counteract generalized disturbances that exist in the pump rate and gap, we augment the disturbance compensation ability as well as enhance the robustness of the slot coating process. In this paper, the following three principal contributions have been outlined: 

\begin{enumerate}
\item
An analysis of the slot coating process is presented, focusing on performance effects due to perturbations in the gap and pump rate.

\item
An optimal tracking control strategy is designed, and the DO-based optimal framework is utilized to approximate the optimal solution of the derived dynamic equation.

\item
To counteract the impact of intricate uncertainties and effectively manage disturbances, this study introduces a generalized disturbance observer-based approach, achieved through the integration of a disturbance compensator in the output channel.

\item
In contrast to the existing literature that discusses the influence of varying gap and pump rate, this paper actively controls the process with modern control theories and enhances robustness against a broader array of uncertainties. To the best of the authors' knowledge, this is the first paper to bridge the connection between the slot coating industry and active disturbance rejection areas.
\end{enumerate}

The remainder of the paper is organized as follows: Section \ref{section2} presents the analysis of the coating flow process, along with an analysis of the experimental platform and its mechanism. Section III shows the mathematical formulation of the dynamic coating process, and the complete DO-based optimal output tracking architecture. In Section IV, illustrative results are shown to compare the performance of the systems to demonstrate the effectiveness of the proposed control architecture, based on a data-driven model under diverse circumstances. Finally, Section V provides the conclusion and outlines future work.

\section{Problem formulation} \label{section2}

\subsection{Description of experimental equipment}
\begin{figure*}
	\centering
	\includegraphics[width=120mm]{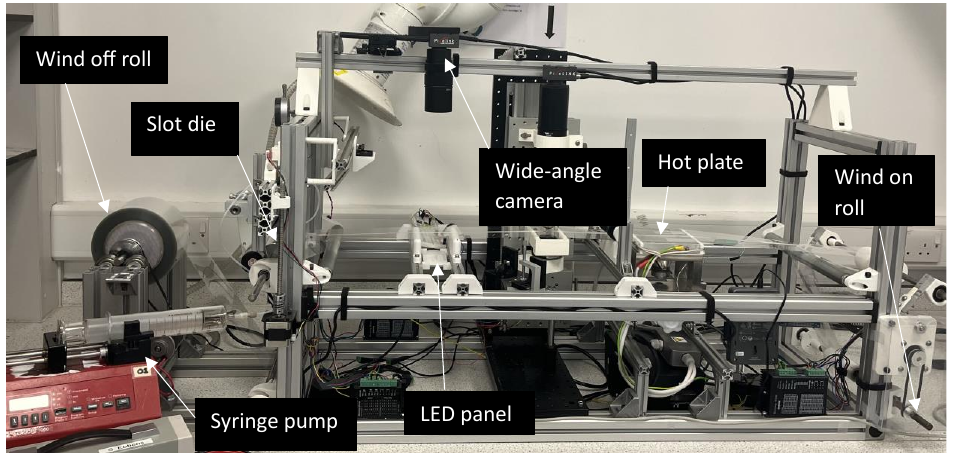}
	\caption{Image of experimental roll-to-roll slot die coating platform.}
	\label{realslodie}
\end{figure*}
This section describes the experimental hardware used in this paper and emphasises the input and output data available that is used for modelling and control. 

The in-house developed roll-to-roll slot die coater is displayed in Fig. \ref{realslodie}. A roll of polyethylene terephthalate film (PMX727/36, HiFi) is unwound, coated by a slot die head (10 mm coating width, InfinityPV), dried by a hot plate and then rewound onto a second roll \cite{UKACC2024abstract}.

The roll-to-roll slot die coater has three computer-controlled inputs: substrate velocity, coating gap, and pump rate. These inputs are controlled by a LabVIEW program in conjunction with a series of stepper motors and a syringe pump (Aladdin). The syringe pump was used in conjunction with a 6 mL syringe (HSW). The LabVIEW program possesses the capability to execute a predefined list of experiments in which these inputs are under computer control.

Wide-angle optical images of the coated region are taken during experiments, with each image capturing a 15-mm-long coated region. The camera (PL-D732CU-T, PixeLINK) is 293 mm away from the slot die and is triggered by the LabVIEW program so that the entire coating is captured. The images are captured at a sample rate of 1.21 - 4.83 seconds, where the program takes into account the speed of the substrate and location of the camera related to the slot die head to ensure there’s no overlap between captured images and that the image is labeled with the coating conditions at the time of printing. 

Subsequently, image analysis is performed to determine the characteristics of the film. First, the green color plane is extracted from the image. An automatic threshold detection is then conducted to separate the coating from the substrate. The mean “grey” pixel intensity value of the coated region is determined. As the coated film is backlit by an LED panel, the mean grey value measured by the camera is dependent on the transmission of light through the coating and is correlated to wet coating thickness. 

The standard deviation of the grey values is also measured by the image analysis. The image analysis also gives the width of the coating using the width of the threshold region and a coefficient of 0.0126 mm/pixel.

All experiments use a 5wt\% titanium dioxide nanopower with primary particle size of 21 nm  and 5wt\% polyvinylpyrrolidone ethanol-based coating solution (all Sigma Aldrich), which was stirred overnight before use. The composition of the coating means it should exhibit similar behavior to many industrially relevant coatings. The principle structure of the test rig is illustrated in Fig. \ref{slotdie}.

\begin{figure*}
	\centering
	\includegraphics[width=130mm]{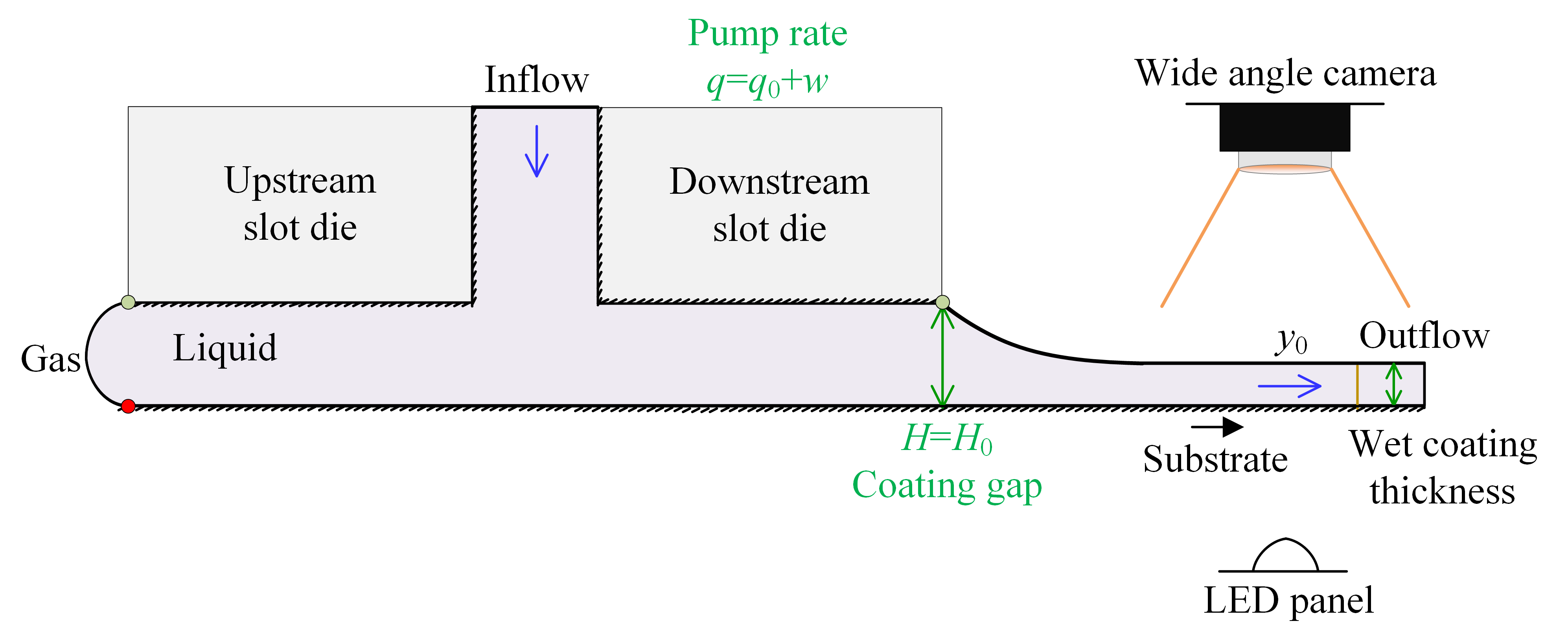}
	\caption{Side-view of slot die coating process, showing pinned meniscus and wide-angle camera.}
	\label{slotdie}
\end{figure*}

\subsection{Disturbances in roll-to-roll slot die coating}
A significant challenge in roll-to-roll slot die coating is to achieve and maintain uniform film deposition, a task heavily reliant on the stability of the operation conditions. In practical roll-to-roll slot die coating systems, it is common to encounter uncertainties or oscillations within the process which affect process parameters through different channels. These uncertainties, often referred to as mismatched or unmatched disturbances \cite{chen2009disturbance}, present significant difficulties in direct compensation. For instance, when the pump rate is fixed and the coating gap is adjusted, any oscillation of the pump rate is problematic for precise control of wet coating thickness. 

 Many pumps often exhibit periodic disturbances; peristaltic pumps, in particular, involve periodic pulses in the fluid flow rate due to the rotation of the lobes. Syringe pumps involve the motion of a syringe plunger, in turn resulting in the solution being displaced and forced through the slot die head, with the fluid flow rate depending on the size of the syringe and speed of plunger movement. The motion of the plunger is often driven using stepper motors, which move in pulses, resulting in a periodic stop-start movement. This movement in turn gives fluid pressure and flow oscillations \cite{zeng2015characterization}, and hence variation in coating thickness. The syringe itself can also give rise to periodic disturbances due to the stick-slip behavior of the plunger within the main body of the syringe. Capes \textit{et al.} observed boluses at regular intervals when using a syringe pump \cite{capes1995fluctuations}. These fluctuations in fluid flow can also disrupt the coating bead leading to lower coating thickness uniformity. 

Periodic disturbances to the coating gap can arise from the presence of an unbalanced roller underneath the slot die. This oscillation disturbs the coating bead, disrupting it from a steady state, leading to sinusoidal variations in coating thickness \cite{lee2015analysis}. This, in turn decreases coating uniformity and leads to variation away from target values. 

These disturbances are often impractical or impossible to remove completely from the equipment.

Therefore, it is crucial to investigate the impact of these disturbances on the coating thickness. This aspect is essential because the thickness of the product is not only an important quality factor, but also influences the performance of the material being coated. For example, coat weight in lithium-ion batteries influences the energy density and power density of the electrode \cite{Reynolds2021Review}. Whilst, higher coat weights are desirable to increase energy density, they typically lower the power density. Therefore, accurate control of coating thickness is desirable to satisfy the power density requirements and maximize the energy density for lithium-ion batteries. Therefore, minimizing variation in coating thickness is crucial for reducing rejection rate and maximizing the output of industrial lines \cite{Grant2022Roadmap}. Hence, strict coating uniformity requirements are introduced for many roll-to-roll slot die coating lines. Because the uniformity requirements of films are demanding, the effect of the previously mentioned disturbances can be critical \cite{lee2015response}. 

Numerous studies have investigated the coating bead's sensitivity to various disturbances and the flow's response to them \cite{romero2008response}. However, very few papers discuss or propose methods for actively improving the robustness of the slot coating system. Therefore, an active DOBOTC system is proposed to handle the aforementioned challenges, a numerical system plant will be provided in Section \ref{sec3}. Consequently, achieving asymptotic stability in the presence of such mismatched uncertainties poses a considerable challenge in the control designs.

Regarding this problem, effective tracking control can be employed to achieve optimal control over the coating process. The optimal controller would be designed to minimize the deviation of key process variables, such as film thickness and width, from their desired values. This involves setting up a cost function that penalizes deviations in these parameters and using the optimal control algorithm to compute the control actions that minimize this cost. Therefore, by ensuring precision in the coating process, optimal tracking control is able to optimize material waste and film quality, which are critical factors in cost-sensitive industrial operations. However, investigations into the deployment of active tracking control within slot die coating systems are markedly under-represented in the current literature.

\section{Active DOBOTC system} \label{sec3}

In this section, a data-driven dynamic model will be methodically developed. Initially, a state space model will be given, serving as a foundational plant for the control framework design. Subsequently, an augmented system will be established, accompanied by the introduction of a tailored cost function for the integration of tracking controllers. Following this, linear optimal control will be implemented as the baseline controller to realize a specific output for the process. Then, we provide the design of a disturbance observer, aimed at identifying disturbances within the system. This will pave the way for the introduction of a generalized scheme for disturbance rejection.

\subsection{Mathematical formulation of the problem}
Suppose an empirical slot coating system is given by
\begin{equation} \label{system}
	\left\{\begin{aligned}
		\dot{x} =& Ax + B_uu + B_d{\omega},\\
		y_{o} =& C_{o}x,
	\end{aligned}\right.
\end{equation}
where $x\in\mathbb{R}^n, {\omega}\in\mathbb{R}^{m}, u\in\mathbb{R}^{m},  y_{o}\in\mathbb{R}$ refer to the system states of the dynamic system, total uncertainties of the slot coating process, control input such as gap or pump rate, and controlled output which refers to the thickness or the width of the film, respectively. $A\in\mathbb{R}^{n \times n}, B_u\in\mathbb{R}^{n\times m} $ and $B_d\in\mathbb{R}^{n\times m}$ are system matrices, while $C_o\in\mathbb{R}^{1\times n}$ represents the output matrix.

Consider a desired output system as:  
\begin{equation}
\begin{aligned}
    \dot{v}(t) &= G v(t), \quad v\left(t_0\right) = v_0, \\y_d(t) &= H v(t),
\end{aligned}
\end{equation}    
where $v \in \mathbb{R}^p$ is the output state vector, $y_d \in \mathbb{R}$ is the desired output. $G$ and $H$ are the matrices of appropriate dimensions. The tracking error is expressed as: 
\begin{equation}\label{outputerror}
   e=y_o-y_d. 
\end{equation}
An optimal control approach can be employed to achieve output tracking by formulating the cost function as follows:
\begin{equation}\label{costfunction}
J=\int_0^{\infty}\left(e^{\mathrm{T}} Q_e e+u_c^{\mathrm{T}} R u_c\right) \mathrm{d} t,
\end{equation}
where $u_{c}\in\mathbb{R}$ denotes the approximated optimal control policy, as derived by the tracking algorithm; $Q_{e}\in\mathbb{R}$ and $R\in\mathbb{R}$ are state and control weighting matrices which determine the penalty on the state error and control input respectively. Now augment the system with the output vector \cite{UKACC2024}:
\begin{equation}
    \overline{x}=\left[\begin{array}{l}
x \\
v
\end{array}\right], 
\end{equation}
so the original cost function (\ref{costfunction}) can be transformed into:
\begin{equation}J=\int_0^{\infty}\left(\overline{x}^{\mathrm{T}} Q \overline{x}+u_c^{\mathrm{T}} R u_c\right) \mathrm{d} t,
\end{equation}
where  
\begin{equation}
Q=\left[\begin{matrix}C_o^{\rm T}Q_eC_o & -C_o^{\rm T}Q_eH\\-H^{\rm T}Q_eC_o & H^{\rm T}Q_eH \end{matrix}\right].
\end{equation}
Thus the augmented system becomes:
\begin{equation}\label{augmented}
\dot{\overline{x}} = \overline{A} \overline{x} + \overline{B}_u \overline{x}u + \overline{B}_d \overline{x}{\omega}.
\end{equation}
where
$$
\overline{A} \overline{x}=\left[\begin{array}{ll}
Ax   \\
Gv
\end{array}\right], \overline{B}_u \overline{x}=\left[\begin{array}{l}
B_u  \\
0
\end{array}\right],
\overline{B}_d \overline{x}=\left[\begin{array}{l}
B_d  \\
0
\end{array}\right].
$$
\begin{remark}
    To apply the proposed mathematical formulation into the slot coating process, a data-driven state space system will be derived in Section IV, by applying the system identification approach to the input-output data collected from images captured by the camera system on the experimental platform, using image processing techniques to extract relevant information.
\end{remark}

\begin{remark}
In the reference system outlined by (2), we create an augmented system alongside the original system plant. This method allows us to showcase the tracking capability in the output channel by confirming the convergence of a specifically built error system. Subsequently, (4) is transformed into a fundamental cost function (6), comprising the augmented system states and control input.
\end{remark}

\subsection{Optimal output tracking system}

For the empirical augmented slot coating system (\ref{augmented}), we introduce a state feedback controller of the form:
\begin{equation}
u_c = -k \overline{x},
\end{equation}
where $u$ represents the feedback control input and $\overline{x}$ denotes the state vector. The dynamics of the closed-loop system with the feedback controller can then be described as:
\begin{equation}\label{eq10}
\dot{\overline{x}} = (\overline{A} - \overline{B}_u k) \overline{x} = \overline{A}_{c l} \overline{x}.
\end{equation}

The objective is to design a controller gain $k$ that optimizes the system performance by balancing control input and system oscillations, as per the cost function defined in (\ref{costfunction}). This is achieved through proper selection of $k$, influencing the eigenvalues of $\overline{A}_{cl}$ for desired system behavior. The basic steps and derivation are given next.

Assuming the existence of a positive definite matrix \( \mathrm{P} \), the following relationship holds:
\begin{equation}
\begin{aligned}
    \frac{\mathrm{d}}{\mathrm{d} t}\left(\overline{x}^T P \overline{x}\right) &= \dot{\overline{x}}^T P \overline{x} + \overline{x}^T P \dot{\overline{x}} \\
    &= -\overline{x}^T\left(Q + k^T R k\right) \overline{x},
\end{aligned}
\end{equation}
where \( \dot{\overline{x}} \) is defined by equation (\ref{eq10}).

Using the following lemma, we can replace \( \dot{\overline{x}} \) and obtain:
\begin{equation}
    \overline{A}^T P + P \overline{A} + Q - k^T \overline{B}_u^T P - P \overline{B}_u k + k^T R k = 0.
\end{equation}

\begin{lemma}
Let \( \overline{A} \), \( \overline{B}_u \), \( Q \), and \( R \) be given matrices with \( Q \) and \( R \) being positive definite. The solution \( P \) to the Algebraic Riccati Equation (ARE):
\begin{equation}
    \overline{A}^P + P \overline{A} - P \overline{B}_u R^{-1} \overline{B}_u^P + Q = 0
\end{equation}
exists and is unique. The optimal gain matrix \( k \) that minimizes the cost function is given by:
\begin{equation}
    k = R^{-1} \overline{B}_u^T P.
\end{equation}
\end{lemma}

By letting \( R = T^T T \), it follows that:
\begin{equation}
\begin{aligned}
    &\overline{A}^T P + P \overline{A} + Q - k^T \overline{B}_u^T P - P \overline{B}_u k + k^T T^T T k = 0.
\end{aligned}
\end{equation}

Now, let us introduce the transformations \( M = -\left(T^{-1}\right)^T \overline{B}_u^T P \) and \( N = Tk \). We can derive:
\begin{equation}
\begin{aligned}
    k^T T^T T k - k^T \overline{B}_u^T P - P \overline{B}_u k &= (M + N)^T(M + N) \\
    &- P \overline{B}_u R^{-1} \overline{B}_u^T P.
\end{aligned}
\end{equation}

Letting \( T k - \left(T^{-1}\right)^T \overline{B}_u^T P = 0 \), we obtain:
\begin{equation}
    k = R^{-1} \overline{B}_u^T P.
\end{equation}

Finally, \( P \) is determined from the Riccati Equation:
\begin{equation}
    \overline{A}^T P + P \overline{A} - P \overline{B}_u R^{-1} \overline{B}_u^T P + Q = 0.
\end{equation}

In summary, the general design steps for the LQR controller are as follows:
\begin{enumerate}
\item 
     Choose the proper weighting matrices $Q$ and $R$.
	\item Solve   $\overline{A}^T P+P \overline{A}-P \overline{B}_u R^{-1} \overline{B}_u^T P+Q=0$ and obtain the matrix ${P}$.
	\item Calculate the control gain $k=R^{-1} \overline{B}_u^T P$ and therefore the baseline optimal control signal for the slot coating system $u_c=-k \overline{x}$.
\end{enumerate}

\subsection{Generalized disturbance rejection scheme}

\subsubsection{DO design}


From the aforementioned analysis, the tracking controller plays a pivotal role in achieving optimal solutions for output tracking problems, ensuring that the desired coating thickness is accurately maintained. However, the attenuation of lumped uncertainties is still needed to deal with the uncertainty of the system. In this section, a DO-based disturbance rejection control is offered to handle the mismatched disturbances, therefore enhancing the robustness of the system.

Suppose an \gls{DO} is proposed for estimating total uncertainties  with the following model \cite{ICITfullpaper}:
\begin{equation}\label{NDO}
	\begin{cases}
		\widehat{\omega}=z+r(x) ,\\
		\dot{z}=-L B_d z-L\left[B_d r+ Ax +B_u u\right],
	\end{cases}
\end{equation}
of which $\widehat{\omega}$ and $z$ represent the estimation of disturbances proposed in (\ref{system}), and respective intermediate states for the proposed \gls{DO}; $r(x)$ refers to the nonlinear observer function to be developed. The observer gain $L$ can be designed as $L=\frac{\partial r}{\partial x}$. The internal states can be easily estimated using a nonlinear observer according to \cite{li2014disturbance}.


\begin{assumption}\label{assumption1}
	The disturbance ${\omega}$ and its derivative $\dot{\omega}$ are bounded, i.e., $||{\omega}||_2\leq H_1$,$||\dot{\omega}||_2\leq H_2$ for any ${k\in \mathbb{N}}$, where $H_1$, $H_2$ are positive constants.
\end{assumption}

\begin{assumption}\label{assumption2}
	Lumped disturbances ${\omega}$ are slowly varying relative to the proposed observer dynamics.
\end{assumption}

\begin{remark}
    In the development of a disturbance observer, it is conventionally assumed that the disturbance exhibits a slower rate of change compared to the observer dynamics. Should this assumption be neglected,  the observer-based system may encounter difficulty in  estimating  the disturbance precisely and executing real-time compensation.
\end{remark}

Given that Assumption \ref{assumption1} and \ref{assumption2} are satisfied, if $L$ is chosen to ensure the stability of the following estimation error system, then $\widehat{\omega}$ will asymptotically converge to ${\omega}$.
\begin{equation}\label{errorsystem}
\dot{e}_0(t) = -L B_d e_0(t),
\end{equation}

Here, the error system is applicable for any $x \in \mathbb{R}^{n}$. $e_0={\omega}-\widehat{\omega}$ denotes the estimation error. Consequently, the developed function $L$, which ensures the asymptotic stability of (\ref{errorsystem}), can be easily selected \cite{li2014disturbance}.

\subsubsection{Composite DOBTOC architecture}
A composite control law of the optimal output stabilization control is derived as:
\begin{equation}\label{compensator}
	u={u}_{c}+{u}_{d}\widehat{\omega},
 \end{equation}
of which ${u}_{d}$ is a compensation input to be designed. It should be noted that the compensator component solely improves the disturbance rejection capability of the composite approach. Additionally, in the absence of lumped uncertainties, the system's global asymptotic stability is maintained with the optimal control input.

Assume $\overline{u}_cx=u_c$, therefore the composite controller can be demonstrated as:
\begin{equation}\label{Compensator}
	u=\overline{u}_cx+{u}_{d}\widehat{\omega},
\end{equation}
where
\begin{equation}\label{compensationgain}
\begin{aligned}
	{u}_{d}=&-[C_{o}(A+B_u \overline{u}_c)^{-1}B_u ]^{-1}\times C_{o}(A\\&+B_u \overline{u}_c)^{-1}B_d .
\end{aligned}
\end{equation}

To handle mismatched disturbances and enhance the robustness of the system, a generalized procedure for designing the \gls{DOBC} can be shown as follows\cite{li2014disturbance}:

\begin{enumerate}
    \item \textbf{Initialize System Parameters:} Define matrices \( A, B_u, B_d, C_o \).
    \item \textbf{Design Baseline Controller:} Ensure closed-loop stability without uncertainties using \( u_c \).
    \item \textbf{Design DO:}
See equation (\ref{NDO}).
    \item \textbf{Verify Assumptions:} Ensure the disturbance \( \omega \) and its derivative are bounded and slowly varying.
    \item \textbf{Design Disturbance Compensation:} Develop compensation gain \( u_d \) and formulate the composite controller (\ref{compensator}).
    \item \textbf{Ensure Stability of Estimation Error System:} Design \( L \) to ensure (\ref{errorsystem}) converges.
    \item \textbf{Implement Composite DOBC:} Integrate the DO-based compensator with the baseline control scheme, ensuring system stability and disturbance rejection capability.
\end{enumerate}

By following these steps, the system's stability and robustness in the presence of disturbances can be effectively ensured. Therefore, the complete architecture of the composite DOBC-based design is depicted in Fig. \ref{figone}.

\begin{remark}
	From the aforementioned DOBC design, consider a standard disturbance compensation law $u={u}_{c}-\frac{\widehat{\omega}}{B_u }$. It should be observed that the compensator is unable to remove the mismatched uncertainties from the system states. An ideal method to handle this difficulty is to counteract the impact of uncertainties in the output with an appropriate compensation function.
\end{remark}

Next this section will prove the system stability as well as the disturbance compensation ability of the complete architecture. First of all, an analysis of the system's BIBO is provided with the proof below.

\begin{theorem}\label{theorem2}
\emph{
	An DOBC architecture composed of the system plant (\ref{system}), DO (\ref{NDO}) and composite controller (\ref{compensator}) is BIBO stable if: i) system plant (\ref{system}) with baseline controller ${u}_{c}$ is globally asymptotically stable without considering the influence of uncertainties; ii) the observer function is selected appropriately to ensure the error is globally asymptotically stable and iii) the compensation function is chosen appropriately to ensure that $P_{d}(x)=B_d+ B_u {u}_{d}$ is continuously differentiable \cite{li2014disturbance}.}
\end{theorem}



\begin{table*}
    \centering
    \caption{Operating conditions and parameters of the test rig}
    \label{tab:my_label}
    \begin{tabular}{ccc}
\hline
Parameter & Symbol &  Value     \\
\hline
Density (g/cm$^3$) & $\rho$ &    0.789    \\
Average water contact angle of substrate ($^{\circ}$)  &  $\theta$  &    66.09 (SD = 4.10)  \\
Web speed (mm/s) & $V$ &     3.1-12.4   \\
Gap height ($\mu$m) & $H_0$ &    100    \\
Target wet coat thickness ($\mu$m) & $y_0$ &    5    \\
\hline
    \end{tabular}

\end{table*}

\begin{proof}
	A subsequent closed-loop system is developed by combining the system equation, control signals and estimation error representation, a subsequent representation of the system in the presence of control signals is summarised as:
		\begin{equation}\label{proof1}
		\begin{cases}
			\begin{aligned}
				\dot{x}=&\left[ Ax  +B_u {u}_{c}\right]-B_u {u}_{d} e\\&+\left[B_d+B_u {u}_{d}\right] \omega, \\
				\dot{e}=&-\frac{\partial r(x)}{\partial x} B_d e .
			\end{aligned}
		\end{cases}
	\end{equation}
with disturbances ${\omega}$ as the input, lumped $x$ and the observer states $e$ as new states for the reconstructed system. Suppose an augmented state $		\overline{x}_{d}=\left[\begin{array}{l}
			x, e		
		\end{array}\right]^{\mathrm{T}}$, therefore
	\begin{equation}\label{newF}
		F(\overline{x}_{d})=\left[\begin{array}{c}
			Ax  +B_u {u}_{c}-B_u {u}_{d} e \\
			-L B_d e
		\end{array}\right] .
	\end{equation}
	Combining (\ref{proof1}) with (\ref{newF}), a reformulated system is represented as:
	\begin{equation}\label{closerepresentation}
		\dot{\overline{x}}_{d}=F(\overline{x}_{d})+\left[\begin{array}{c}
			P_{d}(x) \\
			0
		\end{array}\right] {\omega}.
	\end{equation}
The model $\dot{\overline{x}}_{d}=F(\overline{x}_{d})$ demonstrates asymptotic stability as the first two conditions given in Theorem \ref{theorem3} are satisfied. The disturbance-free system stability in the presence of the baseline control input is demonstrated in Theorem \ref{theorem2}. Following from Lemma 7.1 in \cite{li2014disturbance}, the system (\ref{closerepresentation}) is proven to meet condition (iii) of Theorem \ref{theorem2}.
\end{proof}

Based on the analysis above, one can evaluate the effectiveness of the DO and system stability together with the implemented optimal tracking control architecture. Next consider the analysis of the uncertainty attenuation ability of the composite disturbance compensation system.
\begin{theorem} \label{theorem3} 
\emph{
	For the system (\ref{system}) under DO (\ref{NDO}) and control framework which is composed of the optimal tracking controller and compensation design (\ref{compensator}), the effect of unmatched uncertainties is removed from the output, if uncertainty compensation function ${u}_{d}$ is appropriately designed to ensure that the system (\ref{proof1}) is BIBO stable, while (\ref{compensationgain}) holds \cite{li2014disturbance}.
 }
\end{theorem} 

\begin{proof}
	Consider the system (\ref{proof1}), together with (\ref{Compensator}), we can represent the states as:
 \begin{equation}\label{proof2}
 \begin{aligned}
     x=[&Ax  +B_u \overline{u}_{c}]^{-1}\{\dot{x}-B_u {u}_{d} e\\
     &-\left[B_u {u}_{d}+B_d\right] {\omega}\} .
 \end{aligned}
 \end{equation}
	Combining (\ref{proof2}) with compensation gain (\ref{compensationgain}), gives:
	\begin{equation}
		\begin{aligned}
			y=&{C}_{o}[\overline{A}\overline{x} +B_u \overline{u}_{c}]^{-1} \dot{x}\\&+{C}_{o}\left[ \overline{A}\overline{x} +B_u \overline{u}_{c}\right]^{-1} B_d e.
		\end{aligned}
	\end{equation}
	With the convergence of the system, it can be demonstrated that the effect of uncertainties is removed from the output, as the estimation error and derivatives of the system state exhibit convergence.
\end{proof} 

Therefore, the complete architecture of the composite DO-based-LQR output tracking design is depicted in Fig. \ref{figone}.
\begin{figure}[ht]
	\centering
	\includegraphics[width=85mm]{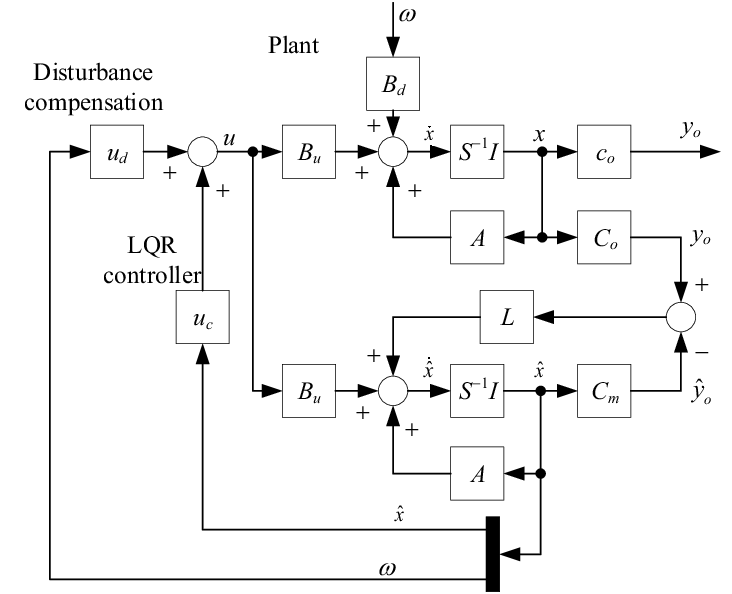}
	\caption{Composite DO-based-LQR output tracking design architecture.}
	\label{figone}
\end{figure}

\section{Simulation results}

This section will compare the goodness of fit for the data-driven model with real data and evaluate the performance of various control algorithms, demonstrating the effectiveness of the proposed methods.

\subsection{Experimental results}
Data for the data-driven model was collected using the roll-to-roll slot die coater and is shown in Fig. \ref{R2R data}. In the test rig, the pump rate and substrate velocity are used to regulate the specific characteristics of the film, such as width and thickness. The inputs are changed with a total of sixteen parameter sets collected, with each parameter being used to coat 120 mm of PET substrate. The input signals are graphically depicted in Fig. \ref{slotdie}.

As mean grey is a measure of transmission of light through the coating and substrate, a calibration of mean grey to theoretical wet coating thickness was performed \cite{Horii2012UltrathinCW,Ramírez2023Blade}. Parameter sets were held for 12 minutes to enable the system to reach a steady state. As mass conservation is observed, the wet coating thickness can be calculated, taking into account the substrate velocity, pump rate, and width of the coating. The average mean grey over the last 6 minutes of each run is plotted against theoretical thickness in Fig \ref{fig:GreyCal} and fits a power shape with an R\textsuperscript{2} value of 0.9895.
A target wet coat thickness of 5 µm was chosen as this thickness has applications in a range of areas from transparent conductive films for touch screens, organic light emitting diodes and organic photovoltaic devices \cite{Wu2013,Shin2016}. A wet coating thickness of 5 µm corresponds to a mean grey value of 167.2.

\begin{figure}
    \centering
    \includegraphics[width=1\linewidth]{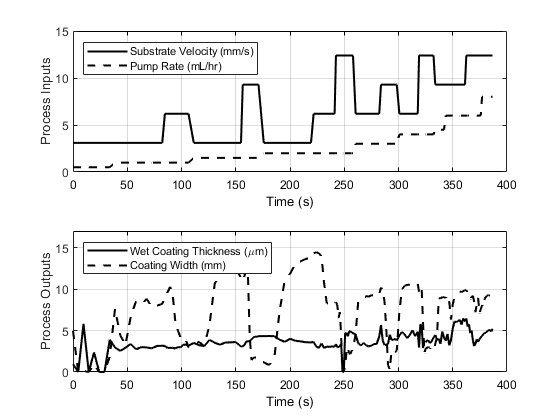}
    \caption{Graph showing process input and process output data from the roll-to-roll slot die coater.}
    \label{R2R data}
\end{figure}

\begin{figure}
    \centering
    \includegraphics[width=1\linewidth]{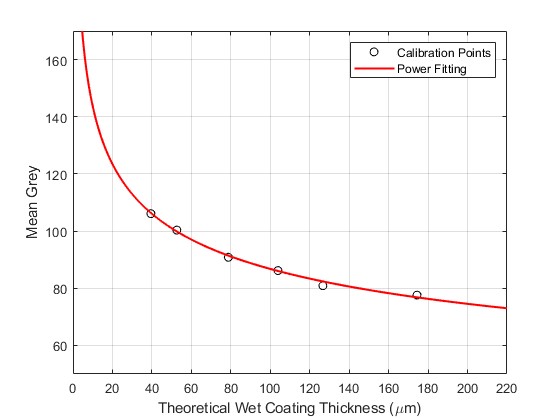}
    \caption{Graph illustrating the calibration of thickness based on mean grey values, showcasing data that conforms to a power-shaped curve. }
    \label{fig:GreyCal}
\end{figure}

\subsection{Dynamic model analysis}

In this paper, the operation parameters of the test rig are shown in Table I. The thickness of the film is continually monitored by employing a camera system to detect the steady-state greyness of the film. The N4SID technique is utilized to construct a data-driven state space model \cite{van1994n4sid}. Five state-space models with different orders are identified, while their respective performance in replicating the dynamics of the system can be visualized in Fig. \ref{si}. The comparison between the data-driven model's fit curve and the original data is also depicted in Fig \ref{si}. This paper employs a second-order system, due to its best-fitting performance among the five models, as evidenced in Fig \ref{so}.  The system matrices are: 
$$
A=\left[\begin{array}{ll}
3.127\quad1.567 \\
0.2803\quad0.258
\end{array}\right], B_u=\left[\begin{array}{l}
6.921\quad0.6338 \\
1.064\quad0.1407
\end{array}\right],
$$
and output matrix
$$
C=\left[\begin{array}{l}
-8.083\quad5.864
\end{array}\right].
$$

The two-channel inputs correspond to the pump rate and gap, respectively. The output refers to the greyness, which reflects the actual steady-state thickness of the film. 

\begin{remark}
The data-driven system receives two inputs: the pump rate and the substrate velocity. It produces two outputs: the film width and the thickness, respectively.
\end{remark}
\begin{figure}[ht]
	\centering
	\includegraphics[width=90mm]{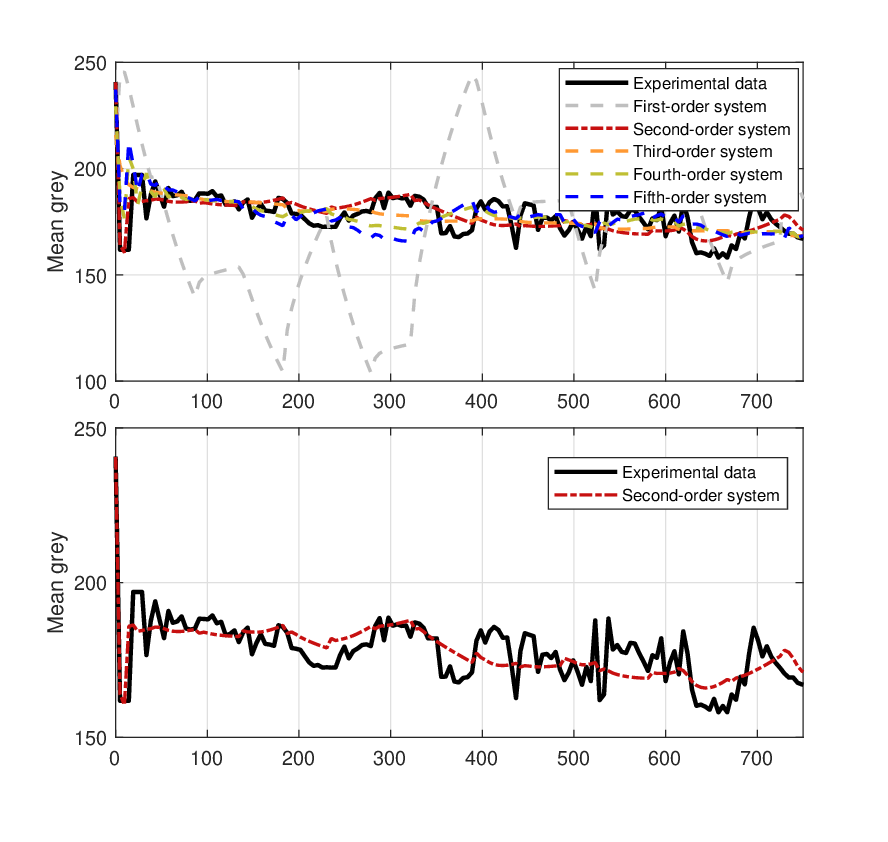}
	\caption{Measured data and simulated model.}
	\label{si}
\end{figure}

\begin{figure}[ht]
	\centering
	\includegraphics[width=90mm]{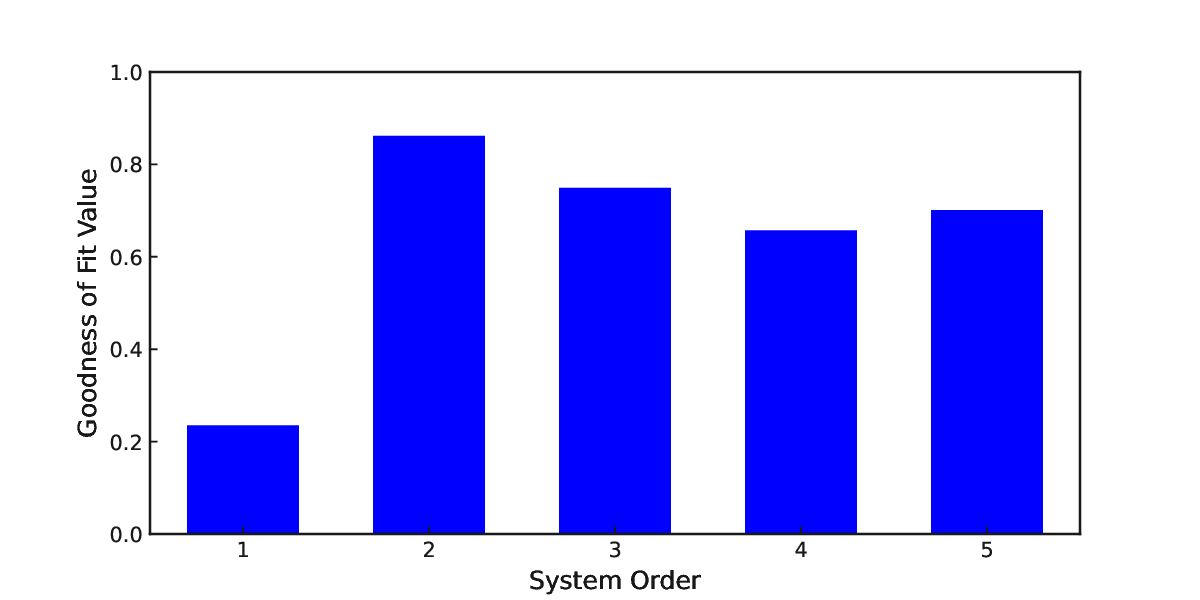}
	\caption{Goodness of fit values for different system orders}
	\label{so}
\end{figure}

\subsection{Matched disturbance rejection}

An output of the data-driven model is defined as $y_d=166.3$, which corresponds to an ideal thickness measured on the test rig, based on the transformation between the greyness and thickness. The choice of initial conditions is $x(0)=[0 \quad 0]^{\mathrm{T}}$. As a result, the trajectories of the model output with the different choices of feedback can be found in Fig. \ref{lqr}, where the three output curves of the slot coating process, represented by the blue lines, can all track the desired trajectories. Three sets of the optimal control parameters are now given as: $\lbrace Q_e = 5, R = 5 \rbrace, \lbrace Q_e = 10, R = 2.5\rbrace, \lbrace Q_e = 10, R = 5\rbrace $. The distinction between those settings may be found in the control input and restriction index, where a lower $R$ indicates a weaker input restriction that will result in higher performance and a greater $Q$ indicates a better control energy input. The optimal performance under particular control settings is indicated by each curve. The values of ideal indexes for each of these settings are shown in Table \ref{J}.

\begin{figure}[ht]
	\centering
	\includegraphics[width=75mm]{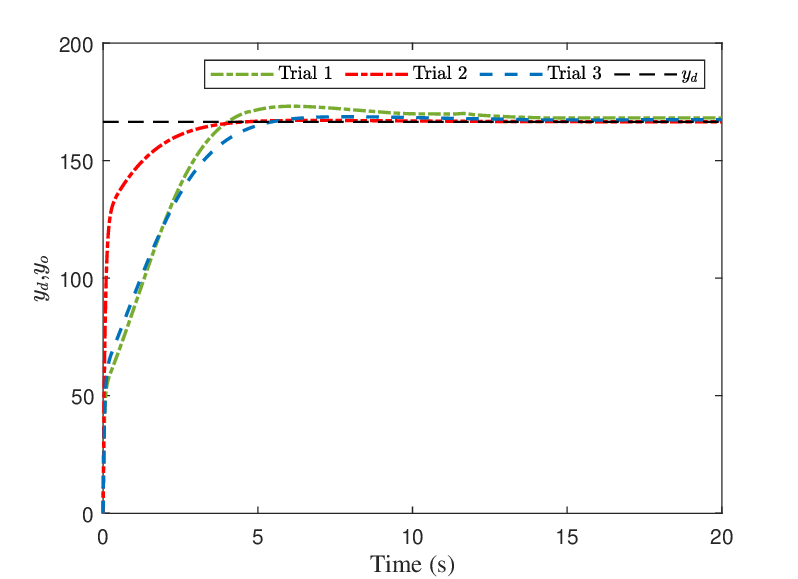}
	\caption{Trajectories of the optimal control with different control and state weighting matrices}
	\label{lqr}
\end{figure}

\begin{table}[ht]
\caption{Value of quadratic cost function}
    \centering
    \begin{tabular}{cm{5cm}<{\centering}}
\hline
Weighting matrices   &  $J$ (from $t=0$ to $t=20$)  \\
\hline
$Q_e=5, R=5$   &  $4.97\times 10^4$  \\
$Q_e=10, R=2.5$  &  $2.32\times 10^5$  \\
$Q_e=10, R=5$  &  $8.16\times 10^4$  \\
\hline
    \end{tabular}
    
    \label{J}
\end{table}

Now assume that there is a periodic oscillation of the coating gap, $\omega=0.2cos(t)$, due to fluctuations in substrate thickness. The selection of the \gls{DO} gain is $L=diag[3 \quad 3]$. The estimation of the periodic perturbation and the estimation error are displayed in Fig. \ref{matchedDO}, where it is evident that the estimation error is asymptotically stable.

\begin{figure}[ht]
	\centering
	\includegraphics[width=90mm]{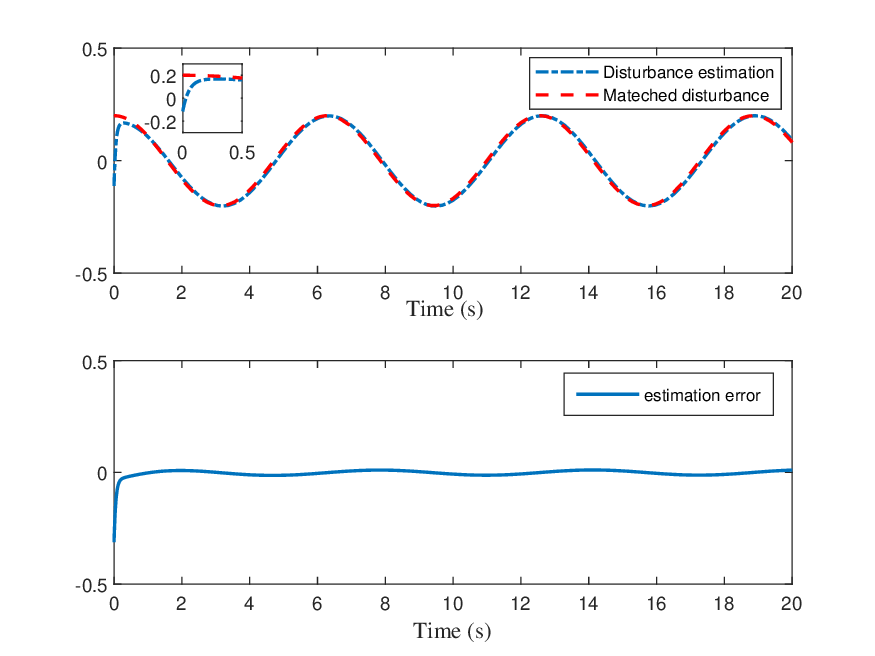}
	\caption{Disturbance estimation and estimation error}
	\label{matchedDO}
\end{figure}

\subsection{Mismatched disturbance compensation}
In the meantime, for the slot coating system, it is not always the case that the disturbance is in the same channel as the control input, for example, the oscillation of the pump rate may be difficult to sense accurately, therefore leading to harm to the quality of the film product. Therefore, other inputs with high precision control can be used to adjust the output in time. Next, assume the desired greyness is also defined as $y_d=166.3$, which corresponds to an optimal value obtained from the test rig. Initial parameters are selected as $x(0)=[0 \quad 0]^{\mathrm{T}}$. The optimal control gains are chosen as $Q_e = 2, R = 1$. The \gls{DO} gain is chosen as $L=diag[3 \quad 3]$. The periodic disturbance that exists in the pump rate is set as $\omega=0.1 \sin(t)$. The system with mismatched input disturbances can be formulated as:

\begin{equation}
	\left\{\begin{aligned}
		\dot{x} =&\left[\begin{array}{ll}
3.127\quad1.567 \\
0.2803\quad0.258
\end{array}\right]x  + \left[\begin{array}{l}
6.921\quad0.6338 \\
1.064\quad0.1407
\end{array}\right](u + {\omega}),\\
		y_{o} =& \left[\begin{array}{l}
-8.083\quad5.864
\end{array}\right]x,
	\end{aligned}\right.
\end{equation}

One approach to address the mismatched disturbance issue involves regulating other inputs to counteract disturbances and deliver the required thickness. In this experiment, we compensate for the mismatched disturbance by tuning the gap in the dynamic optimal output. 

Figures \ref{tracking} - \ref{error} provide illustrative findings. Specifically, Fig. \ref{tracking} illustrates the trajectories of the output signal in conjunction with the desired output, comparing the scenarios of optimal control alone and optimal control with a mismatched DO-based compensator.  Subsequently, In Fig. \ref{error}, the tracking error for the DOBC system is found to be bounded to the origin as expected. In contrast, the output generated by the sole optimal control system exhibits a sine wave characteristic, indicating its limited capability to efficiently attenuate the influence of disturbances in the output channel. Figures \ref{tracking}-\ref{error} provide evidence of the effective output tracking and disturbance rejection abilities of the proposed approach when implemented on an uncertain system.

\begin{table}[ht]
\caption{Comparison of the proposed schemes}
    \centering
    \begin{tabular}{cm{5cm}<{\centering}}
\hline
Method   &  Capabilities  \\
\hline
DO   &  Disturbance estimation  \\
Optimal tracking control  &  Output tracking  \\
DOBOTC  &  Output tracking, disturbance estimation and rejection, mismatched disturbance compensation  \\
\hline
    \end{tabular}
\end{table}

\begin{figure}[ht]
	\centering
	\includegraphics[width=75mm]{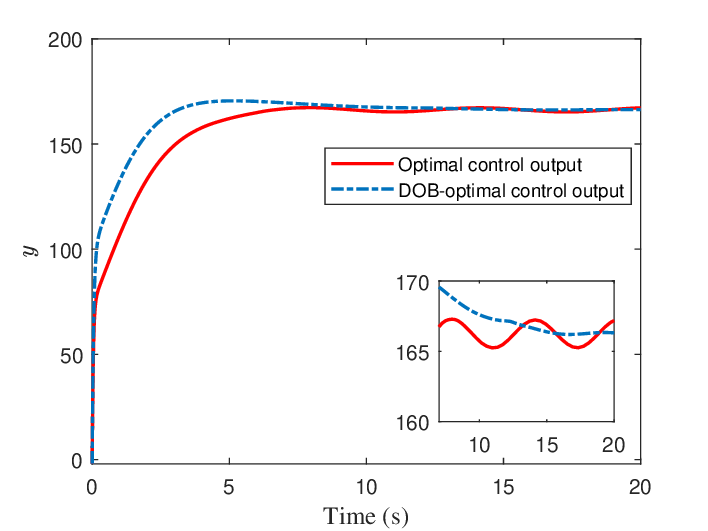}
	\caption{Output trajectories of the optimal control and optimal DOBC system.}
	\label{tracking}
\end{figure}


\begin{remark}
In this analysis, DOBOTC demonstrates significant advantages in managing unpredictable disturbances within roll-to-roll slot die coating systems. DOBOTC's disturbance compensation ability allows for more precise and stable control, which is especially critical in high-precision manufacturing processes. And its ability to enhance system responsiveness and accuracy makes it preferable for applications where high performance under variable conditions is essential. Conventional feedback, while simpler and cost-effective, may lag in environments requiring swift adjustments to disturbances.
\end{remark}

\begin{figure}[ht]\label{statesimage2}
	\centering
	\includegraphics[width=75mm]{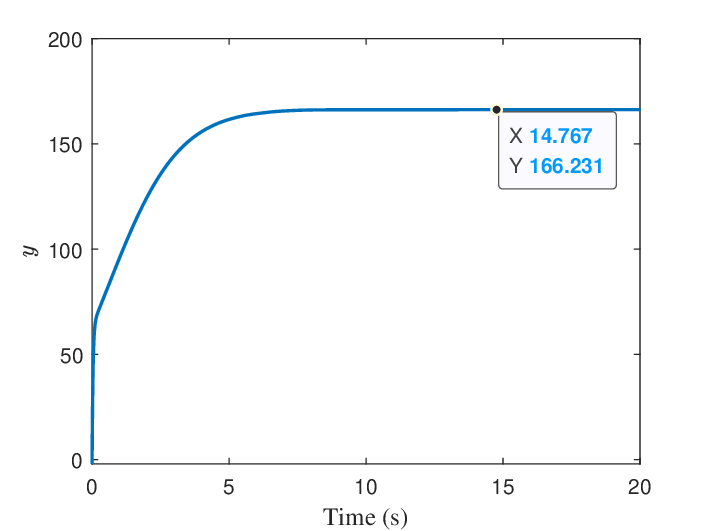}
	\caption{Output trajectory of the optimal tracking control system.}
	\label{fig2}
\end{figure}

\begin{figure}[ht]
	\centering
	\includegraphics[width=75mm]{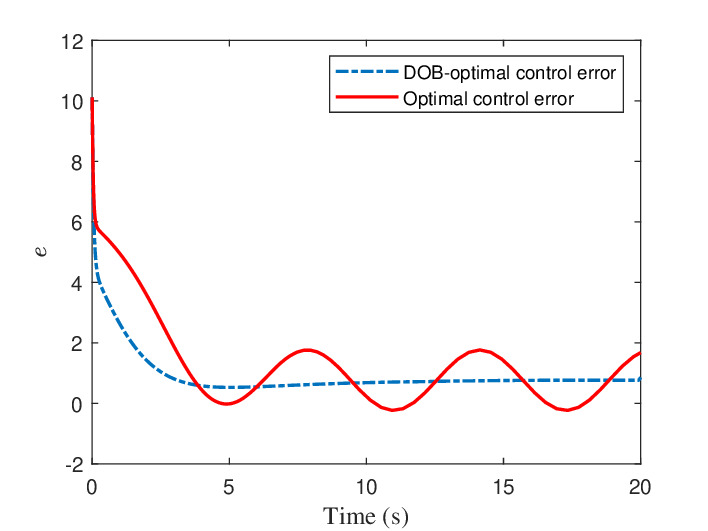}
	\caption{Tracking error of the optimal control and optimal DOBC system.}
	\label{error}
\end{figure}

\section{Implementation on hardware}

Finally, this paper considers the potential to implement this proposed control method on a real world roll to roll coating apparatus. This requires the ability to monitor in real-time the actual disturbances present in the system, without delay. With the current test rig, it is difficult to use sensors that monitor features of the film after deposition, such as those used to supply the data used to build the data-driven model. Even by moving these sensors as close as possible to the slot die head, there will still be a delay before disturbances can be detected, adding complexity to any controller. However, we are currently building a new test rig in which sensors can also be positioned to observe and monitor the film-forming process itself, for example, by direct observation of the coating meniscus. Fig. \ref{meniscus} (a) shows an example image from an additional meniscus camera installed on a roll-to-roll instrument, arranged to directly observe the width of the coating as it is deposited.  Fig  \ref{meniscus} (b) shows the ability to perform image analysis of coating width variations as a function of time, which indeed reveals periodic disturbances of the kind that have been investigated here. Indeed, a frequency analysis revealed these disturbances that could be related to the drum rotation, and pump rate variations, as discussed in the introduction. The availability of this instantaneous process data confirms the viability of the controllers discussed in this study. Therefore, the next goal is implementation of the DOBOTC system on the new test rig. 
Specifically, setting up the DOBOTC system on our new test rig, ensuring that all components are correctly configured and integrated. Then, a series of performance tests will be conducted to assess the system's capabilities and identify the input disturbances. Finally, 
the data collected will be used to test the real-world effectiveness of the DOBOTC algorithms. This will develop better understanding of how well the system performs under practical conditions.

\begin{figure}[ht]
	\centering
	\includegraphics[width=75mm]{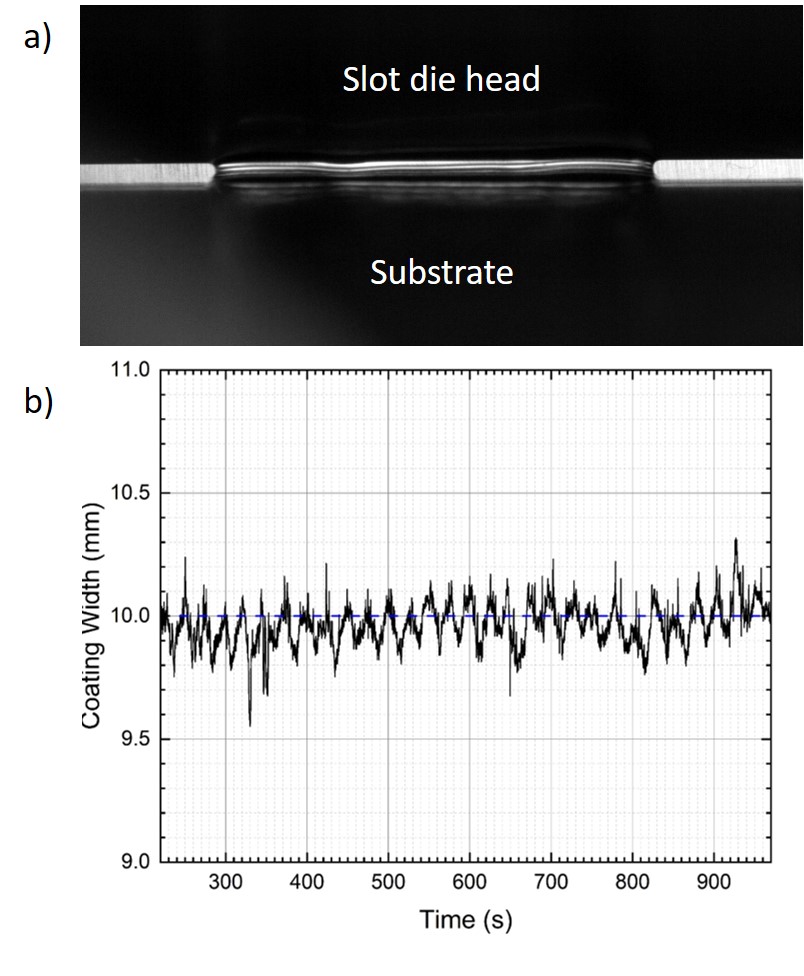}
	\caption{(a) Still frame from the meniscus camera (b) Oscillations in coating width determined from meniscus camera images as a function of time.}
	\label{meniscus}
\end{figure}

\section{Conclusion}
This paper analyses the impact of periodic disturbances on the slot coating system, with a particular focus on coated film thickness. For the dynamic process, the thickness may be sensitive to slight variations in process conditions, which can compromise the consistency of the film thickness. Thus, precise regulation of the thickness, by controlling the flow rate to the die and the gap is essential. Consequently, a composite DO-based optimal output tracking control architecture is demonstrated for coating systems with different types of input disturbances. A generalized DO is introduced to obtain information on the disturbances observed in the gap and flow rate, while most periodic uncertainties can be directly rejected. For the cases where platforms lack the ability to produce precise control signals, or where sensor accuracy is insufficient, the effect of disturbance in the respective channel can be attenuated by a generalized observer-based compensator in the output channel. Furthermore, a DO-based optimal controller is used to pursue an optimized trade-off between control energy input and performance in the dynamic process. The significance of the proposed approach lies in the innovative connection between active disturbance rejection methods and the potentially accurate slot coating process. To the best of our knowledge, this is the first paper to use an active disturbance rejection technique to solve the potential perturbation situation in this area. The simulation results are based on a data-driven model, obtained by the system identification method, from the experimental platform. The analysis, which involves considering the ratio of thickness variation to gap or flow rate disturbances, is given to show the influence of the exposed disturbances in the gap and flow rate, and closed-loop system results illustrate the effectiveness of the presented control architecture. Future research directions include extending the approach to combine with analytical equations which have more explicit physical meanings. Additionally, proving the effectiveness of the real industrial system is valuable, and we have achieved considerable progress, and also demonstrated a method to obtain the required delay free process information.

\bibliographystyle{IEEEtran}
\bibliography{ref.bib}

\end{document}